\documentclass[11pt]{ip-journal} 
\usepackage{graphicx}
\usepackage{amsmath}
\usepackage{amsthm}
\usepackage{amssymb}
\usepackage{mdframed}

\newcommand{\rmd}{\mathrm{d}}
\newcommand{\p}{\partial}
\newcommand{\gb}{\bar{g}}
\newtheorem{theorem}{Theorem}
\newtheorem{proposition}[theorem]
{Proposition}
\newtheorem{lemma}{Lemma}
\newtheorem{corollary}{Corollary}
\newtheorem{definition}{Definition}

\title{The Isoperimetric Problem in \\ Riemannian Optical Geometry}
\author{Henri P. Roesch \& Marcus C. Werner}
\date{\today} 
\begin{document}
\begin{abstract}
In general relativity, spatial light rays of static spherically symmetric spacetimes are geodesics of surfaces in Riemannian optical geometry. In this paper, we apply results on the isoperimetric problem to show that length-minimizing curves subject to an area constraint are circles, and discuss implications for the photon spheres of Schwarzschild, Reissner-Nordstr\"{o}m, as well as continuous mass models solving the Tolman-Oppenheimer-Volkoff equation. Moreover, we derive an isopermetric inequality for gravitational lensing in Riemannian optical geometry, using curve-shortening flow and the Gauss-Bonnet theorem.
\end{abstract}

\maketitle

\section{Introduction}
An astronomically important effect of general relativity is the deflection of light due to the curvature of spacetime, known as gravitational lensing (e.g., \cite{SEF}). When this was first confirmed observationally by Eddington's eclipse expeditions -- whose centenary we celebrate this year -- it provided a crucial corroboration of Einstein's theory. 
\newline
\indent From a mathematical point of view, three geometrical frameworks are usually employed to study this effect\footnote{MCW is grateful to Kyoji Saito for encouraging and regularly attending the {\it Mathematics-Astronomy Seminar Series} (2011-2014) and the {\it Symposium on Gravity and Light} (2013) at Kavli IPMU, University of Tokyo, which was dedicated to such mathematical aspects of gravitational lensing theory. The present collaboration was initiated at the related MRC {\it The Mathematics of Gravity and Light} (2018) organized by the American Mathematical Society.}: null geodesics in 4-dimensional Lorentzian spacetimes (e.g., \cite{P}); the standard approximation of gravitational lensing in 3-space, with applications of Morse theory to image multiplicity and singularity theory to caustics (e.g., \cite{PLW}); and optical geometry, which we will employ in this paper.
\newline
\indent Optical geometry is defined by a 3-dimensional space whose geodesics are spatial light rays (not null geodesics), by Fermat's Principle. More precisely, a stationary spacetime has a timelike Killing vector field, and light rays in 3-space obtained by projecting along this vector field are curves which are geodesic with respect to a Finsler metric of Randers type (cf. \cite{GHWW}). For a static spacetime, the Killing vector field is also hypersurface-orthogonal, so that the spacetime metric may be written
\[
g=g_{00} \rmd t \otimes \rmd t+g_{ij}\rmd x^i \otimes \rmd x^j\,,
\]
thus yielding light rays which are geodesics of a spatial Riemannian metric
\begin{equation}
\gb=-\frac{g_{ij}}{g_{00}}\rmd x^i \otimes \rmd x^j\,,
\label{optmetric}
\end{equation}
called the optical metric. Moreover, if we restrict ourselves to the spherically symmetric case, then light rays will be geodesics in totally geodesic surfaces and we can consider, without loss of generality, the optical metric in the equatorial plane.
\newline
\indent Now it turns out that optical geometry provides a useful framework for gravitational lensing theory yielding, for example, topological criteria for image multiplicity as well as a method to derive the light deflection angle, using the Gauss-Bonnet theorem \cite{GW}. However, one aspect of optical geometry that has hitherto not been exploited, is the fact that the {\it length} of a geodesic in optical geometry, i.e. a light ray, is directly related to {\it time}, by construction (Fermat's Principle), and the time delay between gravitationally lensed images is an important observable.
\newline
\indent Thus in order to approach this problem here, we study geometrical constraints on the lengths of curves bounding areas in the optical geometry of static spherically symmetric spacetimes, in other words, a version of the isoperimetric problem. As is well known (for a detailed review, see e.g. \cite{O}), the length of a closed curve $\p A$ in the Euclidean plane bounding an area $A$ satisfies the standard isoperimetric inequality
\begin{equation}
|\p A|^2\geq 4\pi |A|\,,
\label{isoineq}
\end{equation}
where equality is obtained for {\it circles}, a fact whose discovery has historically been ascribed to Dido, Queen of Carthage\footnote{According to legend, discovered in the process of enclosing the maximum area for the new town centre Byrsa with a string made of hide. Vergil's {\it Aeneis} I 365-369 states, \begin{quote}
{\it Devenere locos, ubi nunc ingentia cernis \\ moenia surgentemque novae Karthaginis arcem, \\ mercatique solum, facti de nomine Byrsam, \\  taurino quantum possent circumdare tergo.}\end{quote} 
}, and is therefore also referred to as {\it Dido's theorem}.
\newline
\indent In section \ref{sec-dido1} of this paper, we begin extending Dido's theorem to the Riemannian optical geometry of static spherically symmetric spacetimes with a constructive proof for the optical geometry of the Schwarzschild solution. Next, we consider the implications of a more general theorem by Bray and Morgan \cite{BM} for Reissner-Nordstr\"{o}m and solutions of the Tolman-Oppenheimer-Volkoff equation. Then in section \ref{sec-ineq}, we proceed from the limiting case of Dido's theorem to derive an isoperimetric inequality applicable to gravitational lensing, where in contrast to (\ref{isoineq}) the enclosed area has Euler characteristic zero. This is established using curve shortening flow and the Gauss-Bonnet theorem.
\newline
\indent The Einstein convention for summation over repeated indices is used throughout this paper. On occasion, we will employ a prime for short to denote differentiation with respect to the radial coordinate. Finally, note also that we set the speed of light to unity, and $c$ denotes a constant radius. 

\section{Dido's theorem}
\label{sec-dido1}

\subsection{Schwarzschild}
We begin by recalling the optical metric of the Schwarzschild solution, in Schwarzschild coordinates $(t,r,\vartheta,\varphi)$. By spherical symmetry, a light ray may be regarded as being located   in the totally geodesic slice $\mathcal{S}$ of the equatorial plane with $\vartheta\equiv \tfrac{\pi}{2}$. Thus, for some $m>0$, we consider the set $\mathbb{R}^2-B(2m)$ with optical metric
\begin{equation}
\gb=\frac{1}{(1-\frac{2m}{r})^2}\rmd r \otimes \rmd r+\frac{r^2}{1-\frac{2m}{r}}\rmd \varphi \otimes \rmd \varphi\,.
\label{optmetric-schw}
\end{equation}
Note also that the Schwarzschild solution has a photon sphere at $r=r_{\rm ph}=3m$, corresponding to closed circular geodesics in the optical geometry.
\begin{theorem}
In the equatorial Schwarzschild optical geometry, the curve $\{r=c\}$ minimizes length within the homology class of piecewise smooth curves bounding the area $|\{r_{\rm ph} \leq r \leq c\}|$ with $\{r=r_{\rm ph}\}$.
\label{theorem1}
\end{theorem}
\begin{corollary} Dido's theorem for Schwarzschild\\
In Schwarzschild optical geometry, light rays bounding solutions of the isoperimetric problem must lie on the photon sphere.
\label{coro-schw}
\end{corollary}
\begin{proof}
Since $\nabla_\varphi\p_\varphi = (3m-r)\p_r$, and only the sets $\{r=c\}$ bound solutions of the isoperimetric problem by Theorem \ref{theorem1}, the result follows for geodesics (light rays) on the photon sphere $\{r=r_{\rm ph}\}$. 
\end{proof}

We now provide a direct constructive proof of Theorem \ref{theorem1}, before exploring the implications of a more general isoperimetric theorem for our optical geometry problem in the next subsection.

\begin{proof}
For the case $c=r_{\rm ph}$, any curve $C$ parametrized by $\lambda$ within the homology class bounding a trivial area with $\{r=r_{\rm ph}\}$ must have $C\supset \{r=r_{\rm ph}\}$, and therefore $\tfrac{r}{\sqrt{1-\frac{2m}{r}}}|_C \geq 3\sqrt{3}m$. So its length satisfies
\begin{align*}
|C|=&\int_C \sqrt{\frac{\dot{r}^2}{(1-\frac{2m}{r})^2}+\frac{r^2\dot{\varphi}^2}{1-\frac{2m}{r}}}\ \rmd \lambda\\ 
\mbox{}\geq &\int_C\frac{r}{\sqrt{1-\frac{2m}{r}}}\dot{\varphi}\ \rmd \lambda \geq 2\pi(3\sqrt{3}m) = |\{r=r_{\rm ph}\}|\,,
\end{align*}
which is the desired result.
\\
\indent
For the case $c>r_{\rm ph}$, we adapt a theorem by Bray \cite{B} and express the optical metric (\ref{optmetric-schw}) in two adjacent domains as follows,
\begin{equation}
\gb=\begin{cases} 
      \omega^2(r)\Big(a^{-2}\nu'_-(r)^2\rmd r\otimes \rmd r+a^2\nu_-(r)^2\rmd \varphi \otimes \rmd \varphi \Big) & (r\leq c)\,, \\
      u^{-2}(r)\nu'_+(r)^2\rmd r \otimes \rmd r+u^2(r)\nu_+(r)^2\rmd \varphi \otimes \rmd \varphi & (r\geq c)\,,
   \end{cases}
\label{metric}
\end{equation}
and as for the matching condition at $r=c$, we require that
\begin{equation}
\omega(c)=1\,, \quad \quad u(c)=a \quad \mbox{such that} \quad \omega'(c)=0\,,
\label{matching}
\end{equation}
so one may regard $u(r)=a=\mathrm{const.}$ for $r\leq c$. We also define a metric $\gb_c$ by setting $\omega\equiv 1$ in (\ref{metric}), such that $\gb=\gb_c$ for $r\geq c$ independent of $\omega$.
\\ 
\indent This form of the metric is motivated by the desire to obtain polar coordinate charts $(R,\varphi)$ with $R(r) = \nu_{\pm}(r)$ in the respective domains. For $r\geq c$, this allows us to recognize the area of any set as precisely the Euclidean area on $\mathbb{R}^2$. For $r\leq c$, $\gb$ is conformally related to a conical $\gb_c$. 
\\
\indent To see this, note first of all that by comparing the metric components of (\ref{optmetric-schw}) and (\ref{metric}) for $r\leq c$,
\begin{align}
a^2\omega^2(r)\nu_-(r)^2 &= \frac{r^2}{1-\frac{2m}{r}}\,, \label{nu-1} \\
a^{-2}\omega^2(r){\nu'_-}^2&=\frac{1}{(1-\frac{2m}{r})^2}\,. \label{nu-2}
\end{align}
Now differentiating (\ref{nu-1}) and evaluating it at $r=c$ gives
\[
a^2\omega(c)\omega'(c)\nu_-(c)^2+a^2\omega^2(c)\nu_-(c)\nu'_-(c)=\frac{c}{1-\frac{2m}{c}}-\frac{m}{(1-\frac{2m}{c})^2}\,,
\]
whence, using the matching conditions (\ref{matching}),
\begin{equation}
a^2(\nu_-\nu'_-)(c)=\frac{c-3m}{(1-\frac{2m}{c})^2}\,.
\label{numinus}
\end{equation}
Also, applying (\ref{matching}) to the positive root of the product of (\ref{nu-1}) and (\ref{nu-2}),
\[
(\nu_-\nu'_-)(c) = \frac{c}{(1-\frac{2m}{c})^{\frac32}}\,,
\]
which together with (\ref{numinus}) implies that
\begin{equation}
a^2=\frac{1-\frac{3m}{c}}{\sqrt{1-\frac{2m}{c}}}<1\,,
\label{a}
\end{equation}
since $1-\frac{3m}{c}<1-\frac{2m}{c}<\sqrt{1-\frac{2m}{c}}$. Now recasting (\ref{metric}) for $r\leq c$ in a polar coordinate chart $(\tilde{R},\tilde{\varphi})$ with $\tilde{R}=\tfrac{\nu_-}{a}$, $\tilde{\varphi}=a^2\varphi$,
\begin{equation}
\gb_c = \rmd\tilde{R}\otimes\rmd\tilde{R}+\tilde{R}^2\rmd\tilde{\varphi}\otimes\rmd\tilde{\varphi}\,, 
\label{cone}
\end{equation}
which, clearly, is locally Euclidean and conical as promised, for (\ref{a}) restricts the range of $\tilde{\varphi}$. This is illustrated in Fig. \ref{fig1}.
\\
\indent Finally, for later reference, we note in passing that comparison of (\ref{optmetric-schw}) and (\ref{metric}) for $r\geq c$ implies
\begin{eqnarray}
\nu_+^2 u^2&=\frac{r^2}{1-\frac{2m}{r}}\,, \label{nu1}\\ 
\nu_+ \nu'_+&=\frac{r}{\left(1-\frac{2m}{r}\right)^{\frac32}}\,. \label{nu2}
\end{eqnarray}
\indent Now given these definitions, we need the following two lemmata to prove the theorem. The first is a technical property of the optical metric expressed as (\ref{metric}); the second is a proof of our statement for the metric $\gb_c$. Then, we shall proceed with the proof of the theorem in general.
\begin{figure}[h]
\centering
\def\svgwidth{0.95\textwidth} 
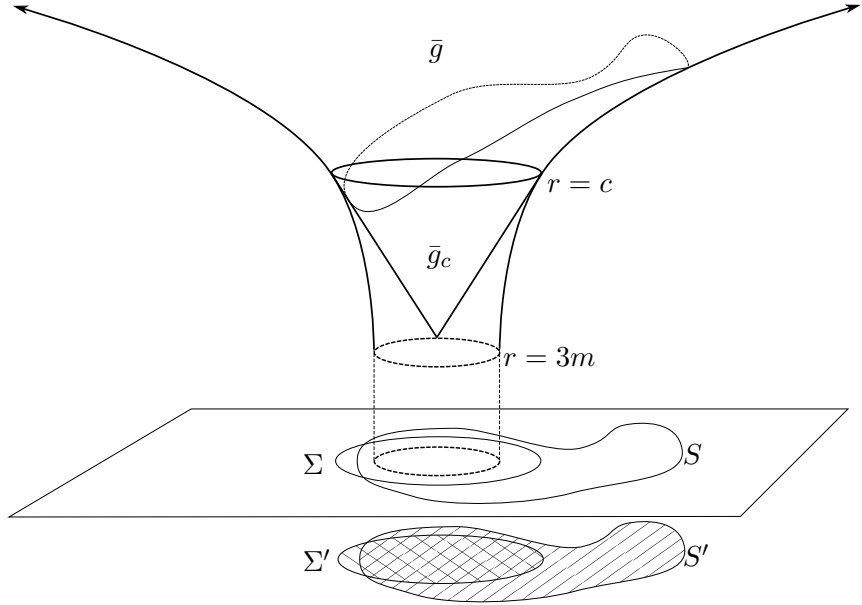
\caption{Schwarzschild optical geometry outside $r_{\rm ph}$. Isometric embedding in Euclidean space, for metrics $\gb$ and $\gb_c$. Note the conical surface for $r\leq c$.}
\label{fig1}
\end{figure}

\begin{lemma} $a^2\leq u^2(r)\leq 1$ for all $r$.
\end{lemma}
\begin{proof}
We begin by recalling that, in the domain $r\leq c$, $u(r)=a <1$ by (\ref{a}). Next, we consider the domain $r\geq c$ and first show that
\begin{equation}
\nu_+^2\geq \frac{r^2}{(1-\frac{3m}{r})\sqrt{1-\frac{2m}{r}}}\,.
\label{inequality}
\end{equation}
To this end, notice that equality obtains at $r=c$ since, from (\ref{nu1}) together with (\ref{matching}) and (\ref{a}),
\[
\nu_+(c)^2=\frac{1}{a^2}\frac{c^2}{1-\frac{2m}{c}}= \frac{c^2}{(1-\frac{3m}{c}) \sqrt{1-\frac{2m}{c}}}\,.
\] 
Now the inequality follows from the fact that the derivative of the left-hand side of (\ref{inequality}) exceeds the derivative of the right-hand side: on the one hand, (\ref{nu2}) implies that
\begin{equation}
\frac{\rmd \nu_+^2}{\rmd r}=2\nu_+{\nu'}_+ = \frac{2r}{(1-\frac{2m}{r})^{\frac32}}\,.
\label{derivative}
\end{equation}
On the other hand, the derivative of the right-hand side is
\begin{align*}
&\frac{\rmd}{\rmd r}\frac{r^2}{(1-\frac{3m}{r})\sqrt{1-\frac{2m}{r}}} = \\
&=\frac{2r}{(1-\frac{3m}{r})\sqrt{1-\frac{2m}{r}}}-\frac{3m}{(1-\frac{3m}{r})^2\sqrt{1-\frac{2m}{r}}}-\frac{m}{(1-\frac{3m}{r})(1-\frac{2m}{r})^{\frac32}} \\
&=\frac{2r}{(1-\frac{2m}{r})^{\frac32}}+\frac{m}{(1-\frac{3m}{r})(1-\frac{2m}{r})^{\frac32}}-\frac{3m}{(1-\frac{3m}{r})^2\sqrt{1-\frac{2m}{r}}} \\
&=\frac{2r}{(1-\frac{2m}{r})^{\frac32}}-2m\frac{1-\frac{3m}{2r}}{(1-\frac{3m}{r})^2(1-\frac{2m}{r})^{\frac32}}\leq \frac{2r}{(1-\frac{2m}{r})^{\frac32}}=\frac{\rmd \nu_+^2}{\rmd r}\,,
\end{align*}
with (\ref{derivative}), yielding the result. We can now prove the inequality of the lemma using an analogous argument: first, recall again from (\ref{matching}) that $u(c)=a$. Then using (\ref{nu1}) and (\ref{inequality}),
\begin{align*}
\frac{\rmd u^2}{\rmd r} &= \frac{\rmd}{\rmd r}\frac{r^2}{\nu^2_+(1-\frac{2m}{r})} = \frac{2r}{\nu_+^2(1-\frac{2m}{r})}-\frac{2m}{\nu_+^2(1-\frac{2m}{r})^2}-\frac{r^2}{\nu^4_+(1-\frac{2m}{r})}\frac{\rmd \nu^2_+}{\rmd r}\\
&=\frac{2r}{\nu_+^4(1-\frac{2m}{r})^2}\left(\left(1-\frac{3m}{r}\right)\nu_+^2-\frac{r^2}{\sqrt{1-\frac{2m}{r}}}\right)\geq 0\,,
\end{align*}
with equality at $r=c$, and therefore $u(r)$ will increase from $a$ for $r>c$. But using (\ref{inequality}), we also see that
\[
\left(1-\frac{2m}{r}\right)\nu_+^2\geq \left(1-\frac{3m}{r}\right)\nu_+^2\geq \frac{r^2}{\sqrt{1-\frac{2m}{r}}}\geq r^2\,,
\]
and hence, again from (\ref{nu1}),
$$u^2(r) = \frac{r^2}{(1-\frac{2m}{r})\nu_+^2}\leq 1\,,$$
which completes the proof of the lemma.
\end{proof}
\begin{lemma}
For $(\mathbb{R}^2-\{0\},\gb_c)$, assume the sets $S'$ and $\Sigma':=\{r\leq c\}$ satisfy $|S'|_c\geq |\Sigma'|_c$. Then, $|\p S'|_c\geq |\p\Sigma'|_c$.
\end{lemma}
\begin{proof}
Recall that the metric $\gb_c$ may be expressed in polar coordinate charts $(R,\phi)\in \mathbb{R}^2-\{0\}$ as, 
\[
\gb_c=\begin{cases} 
      \frac{1}{u(c)^2}\rmd R\otimes \rmd R+u(c)^2R^2\rmd \varphi \otimes \rmd \varphi  & (r\leq c)\,, \\
      \frac{1}{u(R)^2}\rmd R\otimes \rmd R+u(R)^2R^2\rmd \varphi \otimes \rmd \varphi & (r\geq c)\,.
   \end{cases}
\]
Thus, the area element becomes
\[
\rmd A_c=\sqrt{\det \gb_c}\,\rmd R\rmd \varphi=R\rmd R\rmd \varphi=\rmd A_{\mathbb{R}^2}\,,
\]
that is, the standard Euclidean area element, whence
\[
|S'|_c=\int_{S'}\rmd A_c=\int_{S'}\rmd A_{\mathbb{R}^2}=|S'|_{\mathbb{R}^2}\,, \quad |\Sigma'|_c=\int_{\Sigma'}\rmd A_c=\int_{\Sigma'}\rmd A_{\mathbb{R}^2}=|\Sigma'|_{\mathbb{R}^2}\,,
\]
and therefore, by assumption,
\[
|S'|_{\mathbb{R}^2}=|S'|_c\geq |\Sigma'|_c = |\Sigma'|_{\mathbb{R}^2}\,.
\]
Hence, from the Euclidean version of the isoperimetric problem, we conclude that
\begin{equation}
|\p S'|_{\mathbb{R}^2}\geq |\p \Sigma'|_{\mathbb{R}^2}\,. 
\label{isoperi}
\end{equation}
Moreover, the line element of $\gb_c$ satisfies
\begin{equation}
\rmd t_c^2=\frac{1}{u^2}\rmd R^2+u^2 R^2\rmd\varphi^2\geq a^2(\rmd R^2+R^2\rmd \varphi^2)=a^2\rmd t_{\mathbb{R}^2}^2
\label{lineelement}
\end{equation}
since, by Lemma 1,
\[
u^2\geq a^2 \ \mbox{and also} \ 1\geq u^2\geq u^4 \geq a^2 u^2 \Rightarrow \frac{1}{u^2} \geq a^2\,.
\]
Notice also that,
\[
|\p\Sigma'|_c=\int_{\p\Sigma'}\rmd t_c=\int_{\{r=c\}}\rmd t_c=\int_{\{R(c)=\mathrm{const.}\}}a \rmd t_{\mathbb{R}^2}=a|\p \Sigma'|_{\mathbb{R}^2}\,.
\]
Thus, using (\ref{isoperi}) and (\ref{lineelement}),
$$|\p S'|_c\geq a|\p S'|_{\mathbb{R}^2}\geq a|\p\Sigma'|_{\mathbb{R}^2} = |\p\Sigma'|_c\,,$$
as required.
\end{proof}
Now having established the isoperimetric problem for $\gb_c$ with Lemma 2, we shall see in the following how it can be extended to $\gb$, thus proving the theorem. 
\\
\indent First, consider the Gaussian curvature of the equatorial plane in the optical metric. In general, the Gaussian curvature is defined as 
\begin{align}
K&=\frac{\gb\left((\nabla_\varphi\nabla_r-\nabla_r\nabla_\varphi)\tfrac{\p}{\p r},\frac{\p}{\p \varphi}\right)}{\det \gb}=\frac{R_{r\varphi r\varphi}}{\det \gb}\label{gauss1} \\
&=\frac{1}{\sqrt{\det \gb}}\left(\frac{\p}{\p \varphi}\left(\frac{\sqrt{\det \gb}}{\gb_{rr}}\Gamma^\varphi{}_{rr}\right)-\frac{\p}{\p r}\left(\frac{\sqrt{\det \gb}}{\gb_{rr}}\Gamma^\varphi_{r\varphi}\right)\right)\nonumber \\
&=-\frac{2m}{r^3}\left(1-\frac{3m}{2r}\right)\,, \label{gauss2}
\end{align}
which is negative outside the event horizon, $r>2m$. Moreover, recall that in the domain $r\leq c$, the metric $\gb$ is conformally related to $\gb_c$ according to $\gb=\omega^2 \gb_c$. Applying this relation to (\ref{gauss1}), one finds that $K$ is related to the Gaussian curvatures $K_c$ with respect to $\gb_c$ by
\begin{equation}
K=\frac{1}{\omega^2}(K_c-\Delta \ln \omega)\,,
\label{gauss3}
\end{equation}
where $\Delta$ is the usual Laplace-Beltrami operator with respect to $\gb$,
\[
\Delta=\frac{1}{\sqrt{\det \gb}}\frac{\p}{\p x^i}\left(\sqrt{\det \gb} \gb^{ij}\frac{\p}{\p x^j}\right)\,.
\]
Now since $\gb_c$ is conical and locally Euclidean by (\ref{cone}), we have $K_c=0$. But then (\ref{gauss3}) and $K<0$ from (\ref{gauss2}) imply that
\[
\Delta \ln \omega \geq 0 \quad (r\leq c)\,,
\]
and so we know from the Hopf maximum principle that $\ln \omega$ must attain its maximum on the boundary of an annulus $b\leq r\leq c$ for some arbitrary $b$, with a non-zero outward-pointing gradient. Since $(\ln\omega)'(c) = 0$ by the matching condition, this must be at the arbitrary \textit{inner} boundary, $r=b$. Therefore, $\ln \omega$ must increase from its value $\ln\omega(c) = 0$ for $r<c$, and thus we conclude that
\begin{equation}
\omega (r) >1 \quad (r < c)\,.
\label{omega}
\end{equation}
\indent Next, we shall assume that a curve $C$ within the homology class of $\{r=c\}$ bounds, with $\{r=3m\}$, a set $S$ of area $|S|\geq |\Sigma|$ where $\Sigma:=\{3m\leq r\leq c\}$, as illustrated in Fig. \ref{fig1}. Then consider the domain $U = S\cap\Sigma$ and note that, within $\Sigma$ and hence $U$, $r\leq c$ and so $\gb=\omega^2 \gb_c$ and the area element becomes $\rmd A=\sqrt{\det \gb}\,\rmd A_{\mathbb{R}^2}=\omega^2\rmd A_c$. On the other hand, in the domain $S-U$, $\rmd A=\rmd A_c$. Thus,
\[
|S|=\int_S\rmd A=\int_U\omega^2 \rmd A_c+\int_{S-U}\rmd A_c\,,
\]
and so, by assumption,
\[
\int_U\omega^2 \rmd A_c+\int_{S-U}\rmd A_c = |S|\geq |\Sigma| = \int_U\omega^2\rmd A_c+\int_{\Sigma-U}\omega^2\rmd A_c\,,
\]
which together with (\ref{omega}) yields 
\begin{equation}
\int_{S-U}\rmd A_c\geq \int_{\Sigma-U}\omega^2\rmd A_c\geq \int_{\Sigma-U}\rmd A_c\,.
\label{area1}
\end{equation}
By adding the area within $U$ on both sides of (\ref{area1}), we obtain the areas $|S'|_c$ and $|\Sigma'|_c$, respectively, which again by (\ref{area1}) obey
\[
|S'|_c\geq |\Sigma'|_c\,.
\]
We can now apply Lemma 2 to these areas, to conclude that their boundary curves satisfy
\begin{equation}
|\p S'|_c\geq |\p\Sigma'|_c\,.
\label{area2}
\end{equation}
\indent Finally, we turn to the length of the curve $C$. Note that $C$ conists of $\p S-\p U$, its portion outside of $\Sigma$, and $\p U-\p \Sigma$, its portion inside of $\Sigma$. Again, since $\gb=\omega^2 \gb_c$ within $\Sigma$, the line element is $\rmd t=\omega \rmd t_c$. On the other hand, $\rmd t=\rmd t_c$ outside of $\Sigma$. Thus,
\[
|C| =\int_C \rmd t= \int_{\p U-\p\Sigma}\omega \rmd t_c+\int_{\p S-\p U}\rmd t_c\geq \int_C\rmd t_c=|\p S'|_c
\]
using (\ref{omega}). Hence, from (\ref{area2}),
\[
|C| \geq |\p\Sigma'|_c=|\{r=c\}|_c=|\{r=c\}|
\]
since $\gb_c=\gb$ at $r=c$. Overall, therefore, $|S|\geq |\Sigma| \Rightarrow |C|\geq |\{r=c\}|$, completing the proof.
\end{proof}

\subsection{More general case}
Having established Dido's theorem for Schwarzschild optical geometry, we shall now discuss how it may be also be regarded as a consequence of a deeper theorem by Bray and Morgan, which allows a generalization of the result beyond Schwarzschild as well.
\begin{proposition}\cite{BM}(Corollary 2.4)
Given an $n+1$-dimensional hypersurface of revolution with line element
\begin{equation}
\rmd t^2=\rmd r^2+f^2(r)\rmd \Omega^2_n\,,
\label{bm}
\end{equation}
where $\rmd\Omega^2_n$ is the line element of the $n$-dimensional unit sphere, with the following conditions for $r\geq r_0$,
\begin{align}
0\leq &\frac{\rmd f}{\rmd r}<1\,, \label{bm-cond1}\\
&\frac{\rmd f^2}{\rmd r^2}\geq 0\,, \label{bm-cond2}
\end{align}
then every sphere of revolution $S_r$ for $r\geq r_0$ minimizes perimeter uniquely among smooth surfaces enclosing fixed volume with $S_{r_0}$.
\label{prop-bm}
\end{proposition}

\subsubsection{Schwarzschild revisited}
The optical metric of the Schwarzschild equatorial plane (\ref{optmetric-schw}) can be recast as a line element in the form of (\ref{bm}),
\[
\rmd t^2=\frac{\rmd r^2}{\left(1-\frac{2m}{r^2}\right)^2}+\frac{r^2 \rmd \varphi^2}{1-\frac{2m}{r}}=\rmd r^{\ast2}+f(r^\ast)\rmd \Omega_1^2\,,
\]
where $r^\ast$ is known in the physical context as the Regge-Wheeler tortoise coordinate. Thus, comparison yields
\[
\frac{\rmd f}{\rmd r^\ast}\left(r(r^\ast)\right)=\frac{1-\frac{3m}{r}}{\sqrt{1-\frac{2m}{r}}}\,,
\]
and it is immediately apparent that condition (\ref{bm-cond1}) of Proposition \ref{prop-bm} is satisfied outside the photon sphere,
\[
0\leq \frac{\rmd f}{\rmd r^\ast}<1\,: \quad r_{\rm ph}\leq r<\infty\,,
\]
and since
\[
\frac{\rmd^2f}{\rmd r^{\ast2}}\left(r(r^\ast)\right)=\frac{2m}{r^2\sqrt{1-\frac{2m}{r}}}\left(1-\frac{3m}{2r}\right)\,,
\]
likewise condition (\ref{bm-cond2}),
\[
\frac{\rmd^2f}{\rmd r^{*2}}\geq 0\,: \quad 2m<r_{\rm ph}\leq r < \infty\,.
\]
Thus, we recover Dido's theorem for Schwarzschild, Corollary \ref{coro-schw}.

\subsubsection{Reissner-Nordstr\"{o}m}
Next, we shall turn to the Reissner-Nordstr\"{o}m solution of the Einstein-Maxwell equations, described by a mass parameter $m$ and a charge parameter $q$. The line element of the equatorial plane in the optical geometry is 
\begin{equation}
\rmd t^2=\frac{\rmd r^2}{\left(1-\frac{2m}{r^2}+\frac{q^2}{r^2}\right)^2}+\frac{r^2 \rmd \varphi^2}{1-\frac{2m}{r}+\frac{q^2}{r^2}}\,.
\label{optmetric-rn}
\end{equation}
As a result of the additional parameter, Reissner-Nordstr\"{o}m admits {\it two} photon spheres, at radii
\begin{equation}
r^{\pm}_{\rm ph}=\frac{3}{2}\left(m\pm\sqrt{m^2-\frac{8}{9}q^2}\right)\,,
\label{rn-phsph}
\end{equation}
provided that $m^2 > \tfrac{8}{9}q^2$. Now comparing (\ref{optmetric-rn}) with (\ref{bm}), we obtain
\begin{equation}
\frac{\rmd f}{\rmd r^\ast}\left(r(r^\ast)\right)=\frac{1-\frac{3m}{r}+\frac{2q^2}{r^2}}{\sqrt{1-\frac{2m}{r}+\frac{q^2}{r^2}}}\,,
\label{rn-f1}
\end{equation}
and conclude that condition (\ref{bm-cond1}) is indeed satisfied outside of the outer photon sphere at $r^+_{\rm ph}$,
\[
0\leq \frac{\rmd f}{\rmd r^\ast}<1\,: \quad r^+_{\rm ph}\leq r<\infty\,,
\]
and this is also the case for the second condition, (\ref{bm-cond2}), albeit less obviously:
\begin{lemma}
\[
\frac{\rmd^2f}{\rmd r^{*2}}\geq 0\,: \quad r^+_{\rm ph}\leq r < \infty\,.
\]
\end{lemma}
\begin{proof}
Differentiating (\ref{rn-f1}) yields
\begin{equation}
\frac{\rmd^2f}{\rmd r^{\ast2}}\left(r(r^\ast)\right)=\frac{2m}{r^5\sqrt{1-\frac{2m}{r}+\frac{q^2}{r^2}}}P(r;m,q)\,,
\label{rn-f2}
\end{equation}
with the polynomial function of $r$ and the two parameters $m,\ q$,
\begin{equation}
P(r;m,q)=r^3-\frac{3}{2m}(m^2+q^2)r^2+3q^2r-\frac{q^4}{m}\,.
\label{P}
\end{equation}
Now in order to check that $\tfrac{\rmd^2f}{\rmd r^{\ast2}}\geq 0$ for $r^+_{\rm ph}\leq r < \infty$, we first observe in (\ref{rn-f2}) that the right-hand side is positive as $r\rightarrow \infty$, and it remains to be shown that the largest root of (\ref{P}) is at most $r^+_{\rm ph}$. To this end, consider the shifted polynomial $P(r+r^+_{\rm ph};m,q)$. Then it turns out that all monomial coefficients for $m^2>\frac{8}{9}q^2$ are positive. Therefore, by Descartes' Rule of Signs, there is no sign change and thus no positive root of the shifted polynomial, as required. Furthermore, direct computation shows that the limiting case is obtained for
\[
P\Big(r^+_{\rm ph}; m^2=\tfrac{8}{9}q^2\Big)=0.
\]
\end{proof}
Hence, we conclude that Dido's theorem also applies to the outer photon sphere of Reissner-Nordstr\"{o}m.

\subsubsection{Tolman-Oppenheimer-Volkoff}
In the final part of this section, we turn to smooth mass distributions that give rise to static spherically symmetric solutions of general relativity, rather than black hole solutions. Although such models may, in fact, not possess a photon sphere at all, it is instructive to see what the two conditions of Proposition \ref{prop-bm} mean physically in this setting.
\newline
\indent Starting with a general static spherically symmetric spacetime metric,
\[
g=-e^{2A}\rmd t\otimes \rmd t+e^{2B}\rmd r\otimes\rmd r+r^2\left(\rmd \theta\otimes \rmd \theta+\sin^2\theta \rmd \varphi\otimes \rmd \varphi\right)\,,
\]
with functions $A=A(r), \ B=B(r)$, consider a spatial mass density $\rho=\rho(r)$ and pressure $p=p(r)$, which are defined in terms of components of the energy-momentum tensor. Now the cumulative mass parameter of the model is defined by
\[
\mu(r)=4\pi G\int_0^r\rho(\bar{r})\bar{r}^2\rmd \bar{r}\,,
\]
and Einstein's field equations yield
\[
e^{-2B}=1-\frac{2\mu}{r}\,, \qquad \frac{\rmd A}{\rmd r}=\frac{1}{1-\frac{2\mu}{r}}\left(\frac{\mu}{r^2}+4\pi Gpr\right)\,,
\]
as well as the Tolman-Oppenheimer-Volkoff equation of hydrostatic equilibrium,
\[
\frac{\rmd p}{\rmd r}=-\frac{(\rho+p)(\mu+4\pi Gpr^3)}{r^2\left(1-\frac{2\mu}{r}\right)}\,.
\]
Moreover, the corresponding optical geometry has the following line element in the equatorial plane (cf. \cite{GW} for a discussion of lensing properties) whence, again, we can compare with (\ref{bm}), 
\[
\rmd t^2=e^{2B-2A}\rmd r^2+e^{-2A}r^2\rmd \varphi^2=\rmd r^{\ast2}+f^2(r^\ast)\Omega^2_1\,,
\]
and read off
\begin{equation}
\frac{\rmd f}{\rmd r^\ast}\left(r(r^\ast)\right)=e^{-B}\left(1-r\frac{\rmd A}{\rmd r}\right)\,.
\label{tov-f1}
\end{equation}
Thus, the first condition (\ref{bm-cond1}) of Proposition \ref{prop-bm} becomes
\[
0\leq \frac{\rmd f}{\rmd r^*}<1\,: \quad 1-e^B< r\frac{\rmd A}{\rmd r}\leq 1\,,
\]
which can be recast in terms of upper and lower bounds on the pressure gradient,
\begin{equation}
1-\frac{1}{\sqrt{1-\frac{2\mu}{r}}}<\frac{-r}{\rho+p}\frac{\rmd p}{\rmd r}\leq 1\,.
\label{tov-cond1}
\end{equation}
Differentiating (\ref{tov-f1}) yields
\[
\frac{\rmd^2 f}{\rmd r^{*2}}\left(r(r^\ast)\right)=e^{-2B+A}\left(r\frac{\rmd A}{\rmd r}\frac{\rmd B}{\rmd r}-\left(\frac{\rmd A}{\rmd r}+\frac{\rmd B}{\rmd r}\right)-r\frac{\rmd^2A}{\rmd r^2}\right)\,,
\]
which is related to the Gaussian curvature (\ref{gauss1}) of the equatorial plane in the optical geometry,
\begin{align*}
K&=-\frac{1}{f}\frac{\rmd ^2f}{\rmd r^{*2}}\\
&=-\frac{2\mu e^{2A-2B}}{r^3\left(1-\frac{2\mu}{r}\right)^2}\left[1-\frac{3\mu}{2r}-4\pi Gr^3\left(\frac{\rho+p-2\pi Gp^2r^2}{\mu}-\frac{2\rho+3p}{r}\right)\right].
\end{align*}
Therefore, the second condition (\ref{bm-cond2}) can now be expressed as 
\begin{equation}
\frac{\rmd^2 f}{\rmd r^{*2}}\geq 0\,: \quad 1-\frac{3\mu}{2r}-4\pi Gr^3\left(\frac{\rho+p-2\pi Gp^2r^2}{\mu}-\frac{2\rho+3p}{r}\right)\geq 0
\label{tov-cond2}
\end{equation}
in terms of density, pressure and the mass parameter. Thus, (\ref{tov-cond1}) and (\ref{tov-cond2}) provide physical conditions for Dido's theorem to apply, although we shall not pursue a more detailed discussion here. Instead, we proceed beyond the limiting case of Dido's theorem and provide a derivation of an isoperimetric inequality applicable to gravitational lensing in optical geometry, starting with a brief review of curve shortening flow.

\section{An isoperimetric inequality}
\label{sec-ineq}
\subsection{Curve shortening flow}
Suppose that $(\mathcal{S},\gb)$ is a Riemannian surface representing the optical geometry of a static spherically symmetric spacetime.
\begin{definition}
We say a simple closed geodesic $\gamma:\mathbb{S}^1\to\mathcal{S}$ is the boundary of a convexly foliated infinity, provided $\gamma=\partial \Sigma$ for some set $\Sigma\subset \mathcal{S}$, and $\mathcal{S}-\Sigma\cong\mathbb{R}^2-B(1)$. On $\mathcal{S}-\Sigma$ we have,
$$\gb = \rmd r\otimes \rmd r+r^2\rmd\varphi\otimes \rmd\varphi + h$$
with components satisfying $h_{\varphi i}=\mathcal{O}(r), h_{\varphi i,j} = \mathcal{O}(1)$. 
\end{definition}
The reason for Definition 1 is clarified by the following lemma.
\begin{lemma}
The coordinate curves parametrized by $\varphi$ are convex for sufficiently large $r$.
\begin{proof}
It is an easy exercise to show that the unit normal to the coordinate curves $\varphi\mapsto(r_0,\varphi)$ \textit{pointing away from infinity} is given by 
$$N =\frac{1}{\sqrt{\det \gb}}\left(-\gb_{\varphi\varphi}\partial_r+\frac{\gb_{r\varphi}}{\sqrt{\gb_{\varphi\varphi}}}\partial_\varphi\right)\,.$$
It therefore follows that
\begin{align*}
\gb\left( \nabla_{\frac{\partial_\varphi}{\sqrt{\gb_{\varphi\varphi}}}}\frac{\partial_\varphi}{\sqrt{\gb_{\varphi\varphi}}},N\right) &= \frac{-1}{\sqrt{\gb_{\varphi\varphi}\det \gb}}\left(\gb\left( \nabla_{\partial_\varphi}\partial_{\varphi}, \partial_r\right)-\frac{\gb_{r\varphi}}{\gb_{\varphi\varphi}}\gb\left( \nabla_{\partial_\varphi}\partial_\varphi,\partial_\varphi\right)\right)\\
&=\frac{-1}{\sqrt{\gb_{\varphi\varphi}\det\gb}}\Big(\gb_{r\varphi,\varphi}-\frac12g_{\varphi\varphi,r}-\frac12\frac{\gb_{\varphi\varphi,\theta}}{\gb_{\varphi\varphi}} \gb_{r\varphi}\Big)\\
&=\frac{1}{\sqrt{\det \gb}}\left(1+\mathcal{O}\left(\frac{1}{r^2}\right)\right)\,.
\end{align*}
This is clearly positive for sufficiently large $r$.
\end{proof}
\end{lemma}
\begin{definition}
Within a convexly foliated infinity, we say a closed geodesic $\gamma:\mathbb{S}^1\to\mathcal{S}$ is called outermost whenever a closed geodesic $\tilde{\gamma}:\mathbb{S}^1\to\mathcal{S}$ satisfying $r\circ\tilde{\gamma}\geq r\circ \gamma$ implies $\gamma=\tilde{\gamma}$.
\end{definition}
We refer the reader to \cite{A1,A2,G1,G2} for an in-depth study of curve shortening and highlight in the following the main facts needed in our analysis. Curve shortening is given by the flow $C:\mathbb{S}^1\times I \to \mathcal{S}$, $I\subseteq \mathbb{R}$, defined by
$$\frac{\p C_s}{\p s} = \kappa N\,,$$
where $\kappa$ is the geodesic curvature of the curves $C_s$ and $N$ is the unit normal vector (density) field. This is the gradient flow for the length functional maximizing the decrease in length: 
\begin{equation}
\frac{\rmd |C_s|}{\rmd s} = -\int_{C_s}\kappa^2\rmd t\,,
\label{length}
\end{equation}
and the area change with respect to $C_s$ satisfies
\begin{equation}
\frac{\rmd |A_s|}{\rmd s}= -\int_{C_s} \kappa \rmd t\,.
\label{area}
\end{equation}
For any simple closed embedded curve $C_0$ a maximal solution to curve shortening exists on a time interval $0\leq s<s_{\rm max}$. 
\begin{proposition}
For $(\mathcal{S},\gb)$ such that the convex hull of any compact set is compact, if, $s_{\rm max}<\infty$, $C_s$ converges to a point. If, $s_{\rm max}=\infty$, then any tangential derivative of the curvature of $C_s$ satisfies:
$$\lim_{s\to\infty}\sup_{C_s}|\kappa^{(n)}(s)|=0\,.$$
Moreover, any sequence $s_i\to\infty$ has a subsequence $s_{i_j}$ for which $C_{s_{i_j}}$ converges to some geodesic of $(\mathcal{S},\gb)$. In particular, if $(\mathcal{S},\gb)$ has isolated geodesics, then either $C_s$ converges to a point or a geodesic.
\label{prop-csf}
\end{proposition}
\begin{proposition}\cite{H}(Avoidance Principle)
Given any two disjoint curves $C_0$, $\tilde{C}_0$. Under curve shortening the evolving curves $C_s$, $\tilde{C}_s$ remain disjoint throughout the flow.
\end{proposition}

The Avoidance Principle is a very useful property that holds, roughly speaking, for the following reason. If any two curves evolving under curve shortening was to touch tangentially at an instant of time, then at the point of touching, one curvature would have to be greater than or equal the other. This means the two curves have to intersect if we run the flow parameter backwards. Therefore, a first instance of touching is avoided between two initially disjoint curves flowing under curve shortening.
\begin{proposition}\cite{G2}(Corollary 2.6)
The number of inflection points on the curve does not increase with time.
\label{prop-csf2}
\end{proposition}

\subsection{Application to optical geometry}
\begin{lemma}
Assume $\gamma:\mathbb{S}^1\to\mathcal{S}$ bounds a convexly foliated infinity. Then, if $\gamma$ is outermost, all closed geodesics are bounded from infinity by $\gamma$.
\begin{proof}
Assume, up to a possible diffeomorphism of $\mathbb{S}^1$, $r\circ\tilde{\gamma}(\varphi) \geq r\circ \gamma(\varphi)$ for some $\varphi\in\mathbb{S}^1$. For sufficiently large $r$, we can choose a convex curve $\varphi\mapsto(r,\varphi)$ to initiate curve shortening, namely $C_s$. Since both $\gamma_s\equiv \gamma, \tilde{\gamma}_s\equiv \tilde{\gamma}$ under curve shortening we have, by the Avoidance Principle and Proposition \ref{prop-csf}, that some subsequence $C_{s_i}$ converges to a geodesic $C_\infty$ with $r\circ C_\infty\geq \max\{r\circ\tilde{\gamma},r\circ \gamma\}$. Therefore, $C_\infty = \gamma$, and $r\circ\tilde{\gamma}(\varphi)=r\circ \gamma(\varphi)$.
\end{proof}
\end{lemma}

\begin{lemma}
If the geodesic $\gamma$ is outermost, and $C$ is a piecewise smooth closed curve satisfying $r\circ C\geq r\circ \gamma$, then $|\gamma|\leq |C|$. 
\begin{proof}
By smooth curve approximation, it suffices to assume $C$ is smooth. Therefore, flowing $C=C_0$ under curve shortening we know by the Avoidance Principle that any subsequence converging to a geodesic necessarily converges to $\gamma$. Since curve shortening decreases length, the result follows.
\end{proof}
\label{lem-csf}
\end{lemma}
Now this can be applied to gravitational lensing, for instance in the Schwarzschild optical geometry discussed previously, as illustrated schematically in Fig. \ref{fig2}. In this case, the piecewise smooth closed curve $C$ can, of course, be thought of as comprising two geodesics in the optical geometry which correspond to light rays connecting a light source and an observer. 
\\
\indent With this situation in mind, we shall now conclude this paper with a result that adapts the standard isoperimetric inequality (\ref{isoineq}) to this optical geometry context.
\begin{figure}[h]
\centering
\def\svgwidth{0.95\textwidth} 
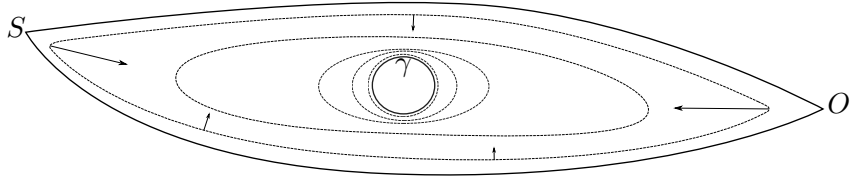
\caption{Curve shortening in optical geometry. Light source and observer (at vertices) are connected by two light rays (bold) enclosing $\gamma$ at $r=r_{\rm ph}$.}
\label{fig2}
\end{figure}

\begin{theorem}
With the hypotheses of Lemma \ref{lem-csf}, if $C$ is convex and the annulus $A$ with $\p A = C\cup \gamma$ supports Gaussian curvature satisfying $K\leq -\delta^2$, then
$$|C|^2\geq |\gamma|^2+\delta^2|A|^2.$$
\end{theorem}
\begin{proof}
We adapt an approach of Topping in \cite{T}. Namely, if $C_0$ is convex then by Proposition \ref{prop-csf2} the curve shortening flow $\{C_s\}$ remains convex for future times, bounding the annuli $A_s\subset A$. By the Gauss-Bonnet theorem and (\ref{area}),
$$\int_{A_s}K\rmd A = -\int_{C_s}\kappa \rmd t = \frac{\rmd|A_s|}{\rmd s}\leq0.$$
As a result,
\begin{align*}
-\delta^2|A_s|\frac{\rmd|A_s|}{\rmd s}&\leq \Big(\int_{A_s}K\rmd A\Big) \frac{\rmd|A_s|}{\rmd s}=\Big(\int_{C_s}\kappa \rmd t\Big)^2\\
&\leq |C_s|\int_{C_s}\kappa^2\rmd t = -|C_s|\frac{\rmd |C_s|}{\rmd s}\,,
\end{align*}
where the last line follows from Jensen's inequality for integrals and (\ref{length}). Now integrating over the flow parameter yields the result.
\end{proof}


\section*{Acknowledgements}
This material is based upon work supported by the National Science Foundation (US) under Grant No. 1641020.

\end{document}